\documentclass[10pt,aps,prl,notitlepage,twocolumn,nofootinbib,superscriptaddress,longbibliography]{revtex4-2}
 
\usepackage{mathtools}
\usepackage{amsmath}
\usepackage[shortlabels]{enumitem}
\usepackage{graphicx,epic,eepic,epsfig,amsmath,latexsym,amssymb,verbatim,color}
\usepackage{amsfonts}       
\usepackage{nicefrac}       
\usepackage{amsmath}
\usepackage{bbm}
\usepackage{float}
\usepackage{tikz}
\usetikzlibrary{chains}
\usetikzlibrary{fit}
\usepackage{pgflibraryarrows}		
\usepackage{pgflibrarysnakes}		
\usepackage{epsfig}
\usetikzlibrary{shapes.symbols,patterns} 
\usepackage{pgfplots}
\usepackage[strict]{changepage}
\usepackage{hyperref}
\hypersetup{colorlinks=true,citecolor=blue,linkcolor=blue,filecolor=blue,urlcolor=blue,breaklinks=true}

\usepackage[marginal]{footmisc}
\usepackage{url}
\usepackage{theorem}

\newtheorem{proposition}{Proposition}
\newtheorem{lemma}[proposition]{Lemma}

\newtheorem{corollary}[proposition]{Corollary}

\def\squareforqed{\hbox{\rlap{$\sqcap$}$\sqcup$}}
\def\qed{\ifmmode\squareforqed\else{\unskip\nobreak\hfil
\penalty50\hskip1em\null\nobreak\hfil\squareforqed
\parfillskip=0pt\finalhyphendemerits=0\endgraf}\fi}
\def\endenv{\ifmmode\;\else{\unskip\nobreak\hfil
\penalty50\hskip1em\null\nobreak\hfil\;
\parfillskip=0pt\finalhyphendemerits=0\endgraf}\fi}
\newenvironment{proof}{\noindent \textbf{{Proof~} }}{\hfill $\blacksquare$}

\newcounter{remark}
\newenvironment{remark}[1][]{\refstepcounter{remark}\par\medskip\noindent%
\textbf{Remark~\theremark #1} }{\medskip}

\newcounter{example}

\mathchardef\ordinarycolon\mathcode`\:
\mathcode`\:=\string"8000
\def\vcentcolon{\mathrel{\mathop\ordinarycolon}}
\begingroup \catcode`\:=\active
  \lowercase{\endgroup
  \let :\vcentcolon
  }

\usepackage{cleveref}
\usepackage{graphicx}
\usepackage{xcolor}

\RequirePackage[framemethod=default]{mdframed}
\newmdenv[skipabove=7pt,
skipbelow=7pt,
backgroundcolor=darkblue!15,
innerleftmargin=5pt,
innerrightmargin=5pt,
innertopmargin=5pt,
leftmargin=0cm,
rightmargin=0cm,
innerbottommargin=5pt,
linewidth=1pt]{tBox}

\newmdenv[skipabove=7pt,
skipbelow=7pt,
backgroundcolor=red!15,
innerleftmargin=5pt,
innerrightmargin=5pt,
innertopmargin=5pt,
leftmargin=0cm,
rightmargin=0cm,
innerbottommargin=5pt,
linewidth=1pt]{rBox}

\newmdenv[skipabove=7pt,
skipbelow=7pt,
backgroundcolor=blue2!25,
innerleftmargin=5pt,
innerrightmargin=5pt,
innertopmargin=5pt,
leftmargin=0cm,
rightmargin=0cm,
innerbottommargin=5pt,
linewidth=1pt]{dBox}
\newmdenv[skipabove=7pt,
skipbelow=7pt,
backgroundcolor=darkkblue!15,
innerleftmargin=5pt,
innerrightmargin=5pt,
innertopmargin=5pt,
leftmargin=0cm,
rightmargin=0cm,
innerbottommargin=5pt,
linewidth=1pt]{sBox}
\definecolor{darkblue}{RGB}{0,76,156}
\definecolor{darkkblue}{RGB}{0,0,153}
\definecolor{blue2}{RGB}{102,178,255}
\definecolor{darkred}{RGB}{195,0,0}

\newcommand{\nc}{\newcommand}
\nc{\rnc}{\renewcommand}
\nc{\lbar}[1]{\overline{#1}}
\nc{\bra}[1]{\langle#1|}
\nc{\ket}[1]{|#1\rangle}
\nc{\ketbra}[2]{|#1\rangle\!\langle#2|}
\nc{\braket}[2]{\langle#1|#2\rangle}

\nc{\proj}[1]{| #1\rangle\!\langle #1 |}
\nc{\avg}[1]{\langle#1\rangle}
\nc{\rank}{\operatorname{Rank}}
\nc{\smfrac}[2]{\mbox{$\frac{#1}{#2}$}}
\nc{\tr}{\operatorname{Tr}}
\nc{\ox}{\otimes}
\nc{\dg}{\dagger}
\nc{\bu}{{\mathbf{u}}}

\nc{\cA}{{\cal A}}
\nc{\cB}{{\cal B}}
\nc{\cC}{{\cal C}}
\nc{\cD}{{\cal D}}
\nc{\cE}{{\cal E}}
\nc{\cF}{{\cal F}}
\nc{\cG}{{\cal G}}
\nc{\cH}{{\cal H}}
\nc{\cI}{{\cal I}}
\nc{\cJ}{{\cal J}}
\nc{\cK}{{\cal K}}
\nc{\cL}{{\cal L}}
\nc{\cM}{{\cal M}}
\nc{\cN}{{\cal N}}
\nc{\cO}{{\cal O}}
\nc{\cP}{{\cal P}}
\nc{\cQ}{{\cal Q}}
\nc{\cR}{{\cal R}}
\nc{\cS}{{\cal S}}
\nc{\cT}{{\cal T}}
\nc{\cU}{{\cal U}}
\nc{\cV}{{\cal V}}
\nc{\cX}{{\cal X}}
\nc{\cY}{{\cal Y}}
\nc{\cZ}{{\cal Z}}
\nc{\cW}{{\cal W}}

\nc{\RR}{{{\mathbb R}}}
\nc{\CC}{{{\mathbb C}}}
\nc{\FF}{{{\mathbb F}}}
\nc{\NN}{{{\mathbb N}}}
\nc{\ZZ}{{{\mathbb Z}}}
\nc{\PP}{{{\mathbb P}}}
\nc{\QQ}{{{\mathbb Q}}}
\nc{\UU}{{{\mathbb U}}}
\nc{\EE}{{{\mathbb E}}}
\nc{\id}{{\operatorname{id}}}

\nc{\supp}{{\operatorname{supp}}}

\newcommand{\set}[1]{ \left\{ #1 \right\} }
\newcommand{\setcond}[2]{ \left\{ #1 : #2 \right\} }

\newcommand{\bellplus}{\ket{\Phi}}
\newcommand{\ghz}[1]{ \ket{\textup{GHZ}_{ #1 }} }
\newcommand{\w}[1]{ \ket{\textup{W}_{ #1 }} }
\newcommand{\ghzlo}[1]{ \cG^{ #1 } } 
\newcommand{\wlo}[1]{ \cW^{ #1 } } 

\newcommand{\esys}{ e } 
\newcommand{\csys}{ c } 
\newcommand{\enode}[1]{ \esys_{ #1 } } 
\newcommand{\cnode}[1]{ \csys_{ #1 } } 

\newcommand{\beginst}{ \ket{\varphi} }

\newcommand{\finalst}[1]{ \ket{\varphi_f^{#1}} }

\newcommand{\recover}[2][]{ V_{#1}^{#2} }
\DeclarePairedDelimiter{\norm}{\lVert}{\rVert}

\usepackage{hyperref}
\hypersetup{colorlinks=true,citecolor=blue,linkcolor=blue,filecolor=blue,urlcolor=blue,breaklinks=true}

\makeatletter
\def\grd@save@target#1{%
  \def\grd@target{#1}}
\def\grd@save@start#1{%
  \def\grd@start{#1}}
\tikzset{
  grid with coordinates/.style={
    to path={%
      \pgfextra{%
        \edef\grd@@target{(\tikztotarget)}%
        \tikz@scan@one@point\grd@save@target\grd@@target\relax
        \edef\grd@@start{(\tikztostart)}%
        \tikz@scan@one@point\grd@save@start\grd@@start\relax
        \draw[minor help lines,magenta] (\tikztostart) grid (\tikztotarget);
        \draw[major help lines] (\tikztostart) grid (\tikztotarget);
        \grd@start
        \pgfmathsetmacro{\grd@xa}{\the\pgf@x/1cm}
        \pgfmathsetmacro{\grd@ya}{\the\pgf@y/1cm}
        \grd@target
        \pgfmathsetmacro{\grd@xb}{\the\pgf@x/1cm}
        \pgfmathsetmacro{\grd@yb}{\the\pgf@y/1cm}
        \pgfmathsetmacro{\grd@xc}{\grd@xa + \pgfkeysvalueof{/tikz/grid with coordinates/major step}}
        \pgfmathsetmacro{\grd@yc}{\grd@ya + \pgfkeysvalueof{/tikz/grid with coordinates/major step}}
        \foreach \x in {\grd@xa,\grd@xc,...,\grd@xb}
        \node[anchor=north] at (\x,\grd@ya) {\pgfmathprintnumber{\x}};
        \foreach \y in {\grd@ya,\grd@yc,...,\grd@yb}
        \node[anchor=east] at (\grd@xa,\y) {\pgfmathprintnumber{\y}};
      }
    }
  },
  minor help lines/.style={
    help lines,
    step=\pgfkeysvalueof{/tikz/grid with coordinates/minor step}
  },
  major help lines/.style={
    help lines,
    line width=\pgfkeysvalueof{/tikz/grid with coordinates/major line width},
    step=\pgfkeysvalueof{/tikz/grid with coordinates/major step}
  },
  grid with coordinates/.cd,
  minor step/.initial=.2,
  major step/.initial=1,
  major line width/.initial=2pt,
}
\makeatother

\usepackage{thmtools}
\usepackage{thm-restate}
\usepackage{etoolbox}
\makeatletter
\def\problem@s{}
\newcounter{problems@cnt}

\newcommand{\allproblems}{\problem@s}
\makeatother

\usepackage{amsmath,amsfonts,amssymb,array,graphicx,mathtools,multirow,bm,times,relsize,booktabs}
\usepackage[T1]{fontenc}
\usepackage{qcircuit}
\usepackage{makecell}
\usepackage{pifont}
\usepackage{bbding}
\usepackage[caption=false]{subfig}

\usepackage[ruled, vlined, linesnumbered]{algorithm2e}

\definecolor{colortwo}{rgb}{0.4,0.77,0.17}
\definecolor{colorthree}{rgb}{0.01,0.51,0.93}
\pgfplotsset{compat=1.18}

\newcommand{\update}[1]{\textcolor{black}{#1}}
\allowdisplaybreaks

\begin{document}

\title{Quantum Entanglement Allocation through a Central Hub}

\author{Yu-Ao Chen}
\thanks{Yu-Ao Chen and Xia Liu contributed equally to this work.}
\affiliation{Thrust of Artificial Intelligence, Information Hub, \\ The Hong Kong University of Science and Technology (Guangzhou), Guangdong 511453, China}
\author{Xia Liu}
\thanks{Yu-Ao Chen and Xia Liu contributed equally to this work.}
\affiliation{Thrust of Artificial Intelligence, Information Hub, \\ The Hong Kong University of Science and Technology (Guangzhou), Guangdong 511453, China}
\author{Chenghong Zhu}
\affiliation{Thrust of Artificial Intelligence, Information Hub, \\ The Hong Kong University of Science and Technology (Guangzhou), Guangdong 511453, China}
\author{Lei Zhang}
\affiliation{Thrust of Artificial Intelligence, Information Hub, \\ The Hong Kong University of Science and Technology (Guangzhou), Guangdong 511453, China}
\author{Junyu Liu}
\affiliation{SeQure, Chicago, IL 60615, USA}
\author{Xin Wang}
\email{felixxinwang@hkust-gz.edu.cn}
\affiliation{Thrust of Artificial Intelligence, Information Hub, \\ The Hong Kong University of Science and Technology (Guangzhou), Guangdong 511453, China}

\date{\today}

\begin{abstract}
Establishing a fully functional quantum internet relies on the efficient allocation of multipartite entangled states, which enables advanced quantum communication protocols, secure multipartite quantum key distribution, and distributed quantum computing. \textcolor{black}{In this work, we propose local operations and classical communication (LOCC) protocols for allocating generalized $N$-qubit W states within a centralized hub architecture, where the central hub node preshares Bell states with each end node. We develop a detailed analysis of the optimality of the resources required for our proposed W-state allocation protocol and the previously proposed GHZ-state protocol.} Our results show that these protocols deterministically and exactly distribute states using only $N$ qubits of quantum memory within the central system, with communication costs of $2N - 2$ and $N$ classical bits for the W and GHZ states, respectively. These resource-efficient LOCC protocols are further proven to be optimal within the centralized hub architecture, outperforming conventional teleportation protocols for entanglement distribution in both memory and communication costs. Our results provide a more resource-efficient method for allocating essential multipartite entangled states in quantum networks, paving the way for the realization of a quantum internet with enhanced efficiency.

\end{abstract}

\maketitle


\paragraph{Introduction.---}

\textcolor{black}{Rapid growth of quantum information science and technologies shows significant potential in the area of computing, communication and sensing. From the computing side, quantum computation marks a paradigm shift in processing power, leveraging quantum mechanics to vastly outperform classical computers~\cite{lloyd1996universal,harrow2009quantum,Childs2010,preskill2012quantum,Montanaro2016,preskill2023quantum,Huang2021b,Chen2024,Bravyi2020b}. On the communication front, the quantum internet represents a revolutionary step in communication technology~\cite{kimble2008quantum,azuma2023quantum,Bennett2002,Leditzky2023,Wang2019b,Li2009}. It extends the capabilities of quantum computing by enabling the transmission of quantum information over long distances~\cite{munro2015inside,muralidharan2016optimal,
wengerowsky2019entanglement,cacciapuoti2019quantum,hu2020efficient,rozpkedek2021quantum,azuma2023quantum,li2024generalized}. This advancement is set to transform secure communication through quantum key distribution and interconnected quantum computing networks, offering unbreakable encryption and enhanced computational capabilities. From the sensing side, quantum technologies might significantly enhance the capability of precision measurement and detect novel quantum phenomena \cite{degen2017quantum,giovannetti2006quantum,giovannetti2011advances}. Central to these advancements is the potential requirement for a centralized hub that could serve as a pivotal node for transmitting quantum data~\cite{liu2023data,liu2024quantum}, quantum private query \cite{giovannetti2008quantum}, blind quantum computing \cite{broadbent2009universal}, and distributed quantum sensing \cite{zhang2021distributed}.}

A pertinent inquiry that emerges in the context of quantum network optimization is the identification of protocols to generate multipartite entanglement between end-nodes with both high fidelity and efficiency. In the realm of entanglement distribution for two user-nodes, protocols encompassing both entanglement generation and entanglement swapping have been shown to be optimal for implementation across quantum networks~\cite{bose1998multiparticle,shi2000optimal,su2016quantum}. Nevertheless, it is important to recognize that the fully functional quantum network extends beyond the two-user paradigm, necessitating the exploration and utilization of multipartite entanglement allocation strategies.

Multipartite entanglement plays a pivotal role in the rapidly evolving fields of quantum networks, quantum information theory, and distributed quantum computing~\cite{Horodecki2024,Streltsov2020,Schwaiger2015,Huber2018a,Regula2018b,Navascues2020,Gour2010}. At the forefront of this fascinating area of study are the Greenberger-Horne-Zeilinger (GHZ) and W states, which serve as quintessential examples of multipartite entanglement. 
The GHZ state~\cite{greenberger1989bell,greenberger1990bell} exemplifies a maximally entangled state that involves multiple particles, highlighting the nonlocal correlations inherent in quantum systems. Its unique property of collapsing entirely upon measurement of one particle makes it an ideal candidate for precision measurements and complex quantum algorithms~\cite{lo2000classical,christandl2005quantum,komar2014quantum}.
Besides, the W state~\cite{Dur2000}, known for its robustness against qubit loss, has been shown to provide an advantage in distributed quantum algorithms such as secure communication~\cite{joo2002quantum,agrawal2006perfect,liu2011efficient,lipinska2018anonymous} and secret voting~\cite{d2004computational}. 
\update{ With their pivotal applications, it then becomes essential to devise efficient protocols for distributing these states across the quantum internet.}

\update{The efficiency of these distribution protocols can be assessed by examining both quantum memory and classical communication cost. Quantum memory cost~\cite{bisio2012memory} is primarily determined by the number of ancillary systems that must remain coherent between successive steps, directly impacting the scalability and practical implementation feasibility of a quantum network. Thus, minimizing its consumption in the protocol is crucial for optimal efficiency.} Classical communication cost, meanwhile, is also vital in this context, as it is closely related to the communication cost of distributed quantum information processing~\cite{lo2000classical, Hayden_2003, harrow2004tight}.

While mature distribution technologies can facilitate the cost-effective sharing of Bell states between two parties~\cite{dai2016generation,yang2020cooling,zhang2023scalable}, the efficient distribution of multipartite entanglements remains a critical and largely unexplored area. Current best-known protocols for accomplishing this task heavily rely on quantum teleportation~\cite{Bennett1993teleportation} as a foundational mechanism~\cite{avis2023analysis, Bugalho2023center}, which is designed to work with general quantum states and is not specifically tailored to the unique characteristics of GHZ and W states.

Given the pivotal role of a centralized hub in facilitating key quantum information processing tasks, it is crucial to investigate the efficient distribution of GHZ and W states within this architecture~\cite{cuquet2012growth, khatri2022design}. The distribution procedure typically begins with various end nodes on the end system initially establishing bipartite entanglement with the central hub. Subsequently, through the application of local operations and classical communication (LOCC) by the central and end system, these bipartite states are transformed into a single multipartite entangled state that encompasses all end nodes.

In this paper, we present deterministic and exact protocols to distribute generalized N -qubit W states via one-way LOCC in a central hub. We further prove the optimality of the resources required by our proposed distribution protocols for the W states, as well as the previous protocols for GHZ states and graph states. From a practical perspective, our approach offers a more resource-efficient strategy in terms of memory cost and communication cost. Contrasting with the best known protocols that demand $2N$-qubit memory cost and $2N$ communication cost for sending messages~\cite{avis2023analysis}, our proposed method significantly reduces the memory requirements to only $N$ qubits in the central system. Furthermore, our results show that these protocols only needs communication cost of $N$ classical bits for the GHZ state distribution and $2N - 2$ classical bits for W state distribution, respectively. From a theoretical perspective, our protocol also highlights a fundamental inequivalence between the GHZ and W states in terms of communication costs for this operational task. 

\begin{figure}[t]
    \centering
    \includegraphics[width=0.6\linewidth]{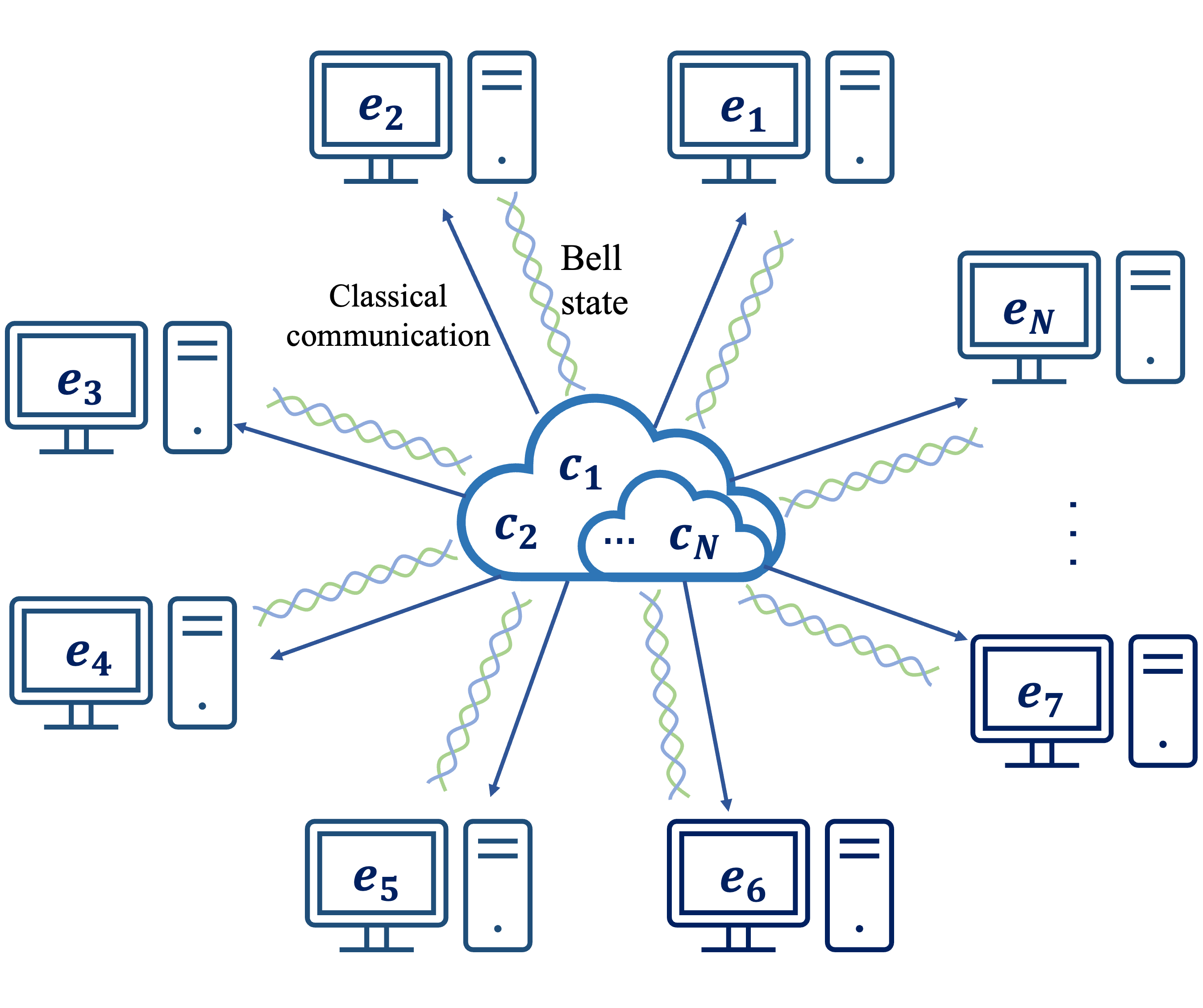}
    \caption{A schematic representation of the central hub considered in this work.
    We specifically study the multipartite entangled state distribution through a central hub using $N$ Bell states. 
     In this framework, each central node $\cnode{i}$ shares a Bell state with each corresponding end node $\enode{i}$. The central system then performs a quantum operation and sends classical information to each end node. Based on the message received, apply local operations on end system. As a result, $N$ end nodes are ultimately distributed a multipartite entangled state.}
    \label{fig:n-center}
\end{figure}

\paragraph{One-way LOCC in a central hub.---}
Consider a scenario wherein a central system $\csys = \cnode{1} \cdots \cnode{N}$ has distributed $N$ Bell states, defined as
\begin{equation}
    \bellplus = (\ket{00} + \ket{11}) / \sqrt{2}
,\end{equation}
to $N$ spatially separated end nodes $\enode{i}$, $i \in \set{1, \ldots, N}$, respectively. Due to physical hardware limitations or security concerns, the end nodes cannot establish classical communication channels with one another. However, the central system preserves the capability to send classical bits to each end node. Under this context, the primary objective of multipartite entanglement allocation is to distribute a target multipartite entangled state across these $N$ end nodes using $N$ preshared Bell states. This configuration presented in Fig.~\ref{fig:n-center} has been experimentally validated in a recent study~\cite{liu2024creation}, which demonstrated the allocation of tripartite entanglement to three spatially distant devices via a central server in metropolitan environments.

The distribution of information between multiple nodes is covered by LOCC, a communication protocol in quantum information theory.  The nodes involved in LOCC protocols are able to perform local operations on their respective systems and may exchange classical bits of information with their neighboring nodes.
Due to the constraints of classical communication in our scenario, we focus on a particular class of LOCC protocols called the one-way LOCC protocol, which has gained significant attention in recent years for its potential to efficiently distribute and process quantum information in a wide range of applications, including quantum computing, quantum communication~\cite{Bennett1993teleportation}, and quantum cryptography~\cite{pirandola2015advances,ouchao2023quantum}. The one-way LOCC protocol in a central hub can be operationally composed of the following steps. 
\begin{itemize}
    \itemsep0em
\item[(i)] $N$ entangled states are distributed jointly on each pair of nodes $\set{\cnode{i}, \enode{i}}, i = 1, \ldots, N$.
\item[(ii)] Using a projection-valued measure $\set{\cM^s}_{s \in \cS}$ labelled by a finite symbol set $\cS$, the central system performs a quantum measurement over central nodes $\cnode{1} \cdots \cnode{N}$ and obtains a measurement outcome \update{$s$}.
\item[(iii)] \update{The central system generates $N$ positive integers, $\alpha_1(s)$, $\ldots$, $\alpha_N(s)$, and sends each of these as a classical message to its corresponding end system.
\item[(iv)] Upon receiving the message, the end system $e_i$ selects and applies the $\alpha_i$-th local recovery operation $\cR^{i,\alpha_i(s)}$.}
\end{itemize}

\noindent Here the operations in step (ii-iv) can be formulated as a quantum operation so-called one-way LOCC operation, defined as
\begin{equation}
    \cL = \sum_{s \in \cS} \cM^s_c \ox \cR^{1,\alpha_1(s)}_{e_1} \ox \ldots \ox \cR^{N,\alpha_N(s)}_{e_N}
.\end{equation}

Within this framework, the central node has quantum correlations with each end node with preshared $N$ Bell states. The goal is to prepare arbitrary target multipartite entangled states in the network via LOCC between the central node and the end nodes.
Throughout this paper, the communication cost of distributing quantum state $\rho$ in a central hub is quantified by the minimum total number of classical bits sent from the central system to end nodes, mathematically defined as
\begin{equation}
C(\rho) = \log \min_{\cL} \setcond{ \prod_{i=1}^N \max_{s \in \cS} \alpha_i(s) }{ \cL(\Phi^{\ox N}) = \rho}  \label{eq:c-cost}
.\end{equation}
Given the complexities associated with practical implementation and noise factor, the projection measurement in step (ii) is considered as some unitary transformations followed by a quantum measurement in computational basis (i.e., $s \in \cS$ is a binary string of length $N$). We further assume that the recovery operation on end nodes are modeled by local unitaries.


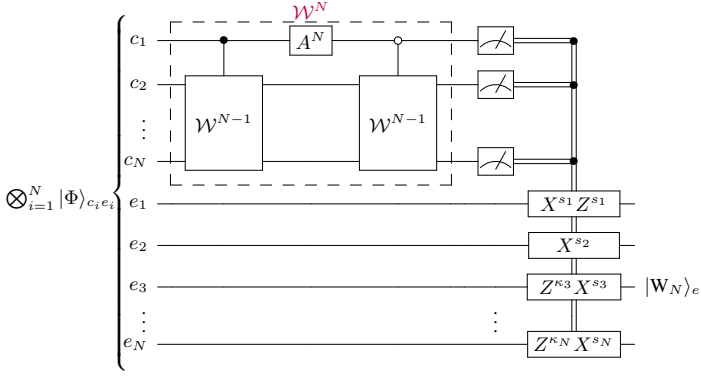
\begin{figure}[t]
\centering
\[
\resizebox{0.8\linewidth}{!}{
\Qcircuit @C=.7em @R=.8em{
    & {} & {} &{} &  \mbox{\textcolor{purple}{$\wlo{N}$}} & {} &{}  & {}  \\
    & \lstick{\cnode{1}} & \qw & \ctrl{1} & \gate{A^N}    &\ctrlo{1}& \qw & \qw &\meter& \control \cw \cwx[1]\\
    & \lstick{\cnode{2}} & \qw &\multigate{2}{\wlo{N-1}} & \qw &\multigate{2}{\wlo{N-1}} & \qw & \qw &\meter& \control \cw \cwx[2]\\
    \vdots & & &  & \nghost{U_{N-1}}  \hspace{2.5em}  & \nghost{\wlo{N-1}}  & & & & & \\
    & \lstick{\cnode{N}} & \qw &\ghost{\wlo{N-1}} & \qw  &\ghost{\wlo{N-1}} & \qw & \qw &\meter& \control \cw \cwx[1]\\
    & \lstick{\enode{1}} & \qw & \qw & \qw  & \qw  & \qw & \qw & \qw & \gate{\hspace{0.3em} X^{s_1} Z^{s_1} \hspace{0.3em}} \cwx[1] & \qw \\
    & \lstick{\enode{2}} & \qw & \qw & \qw  & \qw  & \qw & \qw & \qw & \gate{\hspace{1.1em} X^{s_2} \hspace{1.1em}} \cwx[1] & \qw  \\
    & \lstick{\enode{3}} & \qw & \qw & \qw  & \qw  & \qw & \qw & \qw & \gate{\hspace{0.3em} Z^{\kappa_3} X^{s_3} \hspace{0.3em}} \cwx[1] & \rstick{\w{N}_\esys} \qw  \\
    \vdots &  & &   &  \hspace{2.5em}  &  & & &\vdots & \cwx[1] & \\
    & \lstick{\enode{N}} & \qw & \qw & \qw  & \qw  & \qw & \qw & \qw & \gate{Z^{\kappa_N} X^{s_N}} & \qw
    \gategroup{2}{4}{5}{6}{1.5em}{--}
    \inputgroupv{2}{10}{2.5em}{8em}{\bigotimes_{i=1}^N \bellplus_{\cnode{i} \enode{i}} \hspace{5em}}}
}
\]
\caption{Protocols for distributing generalized $N$-qubit W states for $N \geq 3$. In this setting, each end node $e_i$ preshares $N$ Bell states with the central node $c_i$, and the local operation on each end node $\enode{i}$ depends on the measured outcome $s = s_1 \cdots s_N \in \set{0, 1}^{\times N}$. 
}~\label{fig:n-w}
\end{figure}

\paragraph{$N$-qubit W state allocation.---}
The W state is robust against qubit loss and offers many advantages in various quantum information processing tasks. The generalized $N$-qubit W state is defined as follows: 
\begin{align}
  \w{N} \coloneqq \frac{1}{\sqrt{N}}\sum_{k=1}^N X_k \ket{0}^{\ox N},  
\end{align}
where $X_k$ represents the Pauli-$X$ gate acting on the $k$-th qubit.

The earlier method of distributing W states depended on quantum teleportation, which is a general-purpose method for transmitting any pure state. Since our goal is to distribute a specific known quantum state here, it would be more efficient to use a protocol specifically designed for that state in order to further reduce the required resources. With this in mind, we introduce a more efficient one-way LOCC protocol that attains the same objective. 
Unlike teleportation, which requires $N$ preshared Bell states and a classical communication cost of $2N$ classical bits, 
our proposed protocol achieves the same result with fewer resources. Additionally, by employing the central hub in Fig.~\ref{fig:n-center}, \update{we minimize the memory cost to $N$ qubits, in contrast to the $2N$ qubits needed in other architectures,} thereby enhancing its resource efficiency.

\begin{proposition}\label{prop:n-w}
    There exists a one-way LOCC protocol that deterministically and exactly distributes an $N$-qubit W state in a central hub with $N$ preshared Bell states and an optimal classical communication cost of $2N-2$ classical bits.
\end{proposition}                                        
The existence of such a protocol can be elaborated as follows. The protocol begins by sharing $N$ copies of Bell states between the central nodes and the end nodes $\beginst=\bigotimes_{i=1}^N\ket{\Phi}_{e_ic_i}$. To obtain the desired state, one can perform the $N$-qubit operation $\wlo{N}$,
\begin{align}~\label{eq:w-un}
    \wlo{N} =\frac{1}{\sqrt{N}}\sum_{r=1}^N Z^{\ox (r-1)}\ox X \ox I^{\ox(N-r)}
\end{align}
on the central system followed by a $N$-qubit computational measurement, where $I, X, Z$ are Pauli operators. Depending on the measurement result $s = s_1 \cdots s_N$, the $k$-th end node can correspondingly perform the single qubit recovery unitary
\begin{equation}
\begin{cases}
    X^{s_k} Z^{s_k}, &\textrm{ if } k = 1; \\
    X^{s_k}, &\textrm{ if } k = 2; \\
    Z^{\kappa_k} X^{s_k}, &\textrm{ if } 3 \leq k \leq N
\end{cases}
\end{equation}
for $\kappa_k = \left(\sum_{l=2}^{k-1} s_l\right)\mod 2$. Tracing out the central system would finally obtain the recovered state $\w{N}$, across $N$ end nodes. 
Note that the implementation of $\wlo{N}$ can be constructed recursively by $\wlo{N - 1}$ and one rotation gate $A^N = R_y ( 2 \arccos (\sqrt{(N-1)/N}) )\cdot Z$, with the base case $\wlo{2} = (X \ox I + Z \ox X) / \sqrt{2}$.
We summarize this protocol in Fig.~\ref{fig:n-w}.

Based on the expression of the recovery operation, we observe that except the first end node requires $s_1$, the second end node requires $s_2$, and each subsequent $k$-th $(k\geq 3)$ end node requires both $\kappa_k$ and $s_k$. 
Consequently, the communication cost amounts to $2N-2$ classical bits. The detailed proof of Proposition~\ref{prop:n-w} is provided in the Supplementary Material. \textcolor{black}{Specifically, because of the recovery operation, we expect that such a one-way LOCC protocol might be closely connected to the theory of quantum error correction and detection, similar to algorithms like quantum error filtration that bridges error suppression with quantum communication \cite{gisin2005error,lee2023error}. We leave potential connections between quantum error correction and our LOCC protocol for future research. }

We confirm that the classical communication cost required to allocate the $N$ -qubit W state in a central hub is indeed optimal, as described in the protocol outlined in Proposition~\ref{prop:n-w}, regardless of the operations performed by the central and end systems of the central hub. In fact, for a fixed unitary $\wlo{N}$ in Fig.~\ref{fig:n-w}, there exist several optional local operations that can be performed on the subsystems $\{e_i\}_{i=1}^N$ to obtain $\w{N}$ through classical communication. Likewise, for other optional unitaries $\wlo{N}$, there are also multiple possible local operations that can achieve $N$-qubit $W$ states based on classical communication. However, our protocol establishes that, in both of these scenarios, the optimal classical communication cost to distribute an $N$-qubit W state in a central hub framework is $2N-2$ classical bits. This discovery confirms the optimality of the proposed protocol for the $W$ state distribution in terms of classical communication resources. A detailed proof is provided in the Supplemental Material.

\paragraph{$N$-qubit GHZ states allocation.---}
The GHZ state is another important resource in quantum information processing. The generalized $N$-qubit GHZ state is defined as follows,
\begin{align}
   \ghz{N} \coloneqq \frac{1}{\sqrt{2}}(\ket{0}^{\ox N} + \ket{1}^{\ox N}) .
\end{align}
Refs.~\cite{cuquet2012growth,khatri2022design} have established that their allocation protocols necessitate $N$ preshared Bell states and involve a communication cost of $N$ classical bits. Our analysis in Proposition~\ref{prop:n-ghz} can further show the optimality of these protocols concerning resource efficiency. For completeness, we provide an outline of the one-way LOCC protocol utilizing a central hub below.

\begin{proposition}\label{prop:n-ghz}
There exists a one-way LOCC protocol that deterministically and exactly distributes an $N$-qubit GHZ state in a central hub with $N$ preshared Bell states and an optimal classical communication cost of $N$ classical bits.
\end{proposition}
 
Similarly, we present a concrete protocol for proving the case of GHZ states allocation, which is analogous to the proof of Proposition~\ref{prop:n-w}. To acquire the recovered GHZ state $\ghz{N}$, the protocol involves applying the operation 
\begin{align}
  \ghzlo{N}= \left(H \ox I^{\ox(N-1)}\right) \cdot 
  \begin{pmatrix}
      I^{\ox (N-1)} & 0\\
      0 & X^{ \ox (N-1)}
  \end{pmatrix}
\end{align}
to the central system followed by a computational measurement, after the central and final nodes share N copies of Bell states, where $H$ is the Hadamard gate. The corresponding recovery operation requires the application of a Pauli-$Z$ gate to the first end node and a Pauli-$X$ gate to the remaining end nodes.

Building upon the analysis of W states, \update{the protocol for $N$-qubit GHZ states can be implemented with only $N$-qubit quantum memory cost in the central system.} To achieve the allocation of $N$-qubit GHZ states, we initially require $N$ preshared Bell states. Each end node then requires $s_i$ to determine the operation on $e_i$, as specified by the recovery operation $\recover{s} = Z^{s_1}\ox  \bigotimes_{k=2}^{N} X^{s_k} $. Consequently, the communication cost for this allocation amounts to $N$ classical bits. We leave the detailed proof of Proposition~\ref{prop:n-ghz} in the Supplementary Material.

As we have previously demonstrated, the protocol we proposed for the W state distribution is optimal in terms of classical communication resources. Additionally, we have discovered that the lower bound for the necessary classical bits when distributing any $N$-qubit pure state using centralized hub entanglement allocation protocol is $N$ bits, inspired by Proposition~\ref{prop:n-ghz}. This implies that regardless of the chosen unitary and local operations acting on $\{c_i\}_{i=1}^N$ and $\{e_i\}_{i=1}^N$, the classical bits required to distribute an arbitrary pure state cannot be less than $N$. As a result, the protocol described in Proposition~\ref{prop:n-ghz}, which has a communication cost of $N$ classical bits, can be considered optimal. This confirms the optimality of the communication cost required for the proposed GHZ state allocation protocol in a central hub. The detailed proof can be found in the Supplemental Material.

\paragraph{Discussion.---}
In this work, we investigate the protocols for distributing $N$-qubit generalized W and GHZ states across centralized hub architecture. Our protocols significantly reduce the resource requirements, with $2N-2$ communication cost needed for allocating $N$-qubit W states and $N$ communication cost for $N$-qubit GHZ states. Ref.~\cite{khatri2022design} presents protocols for allocating any $N$-qubit graph state, and we can further demonstrated that it attains the optimal resource requirements in the Supplementary Material. These results are summarized in Table~\ref{tab:resource}. The minimized resource usage and enhanced allocation efficiency of multipartite entangled states could open up promising opportunities for various applications, including quantum communication protocols, secure quantum key distribution, and quantum computing.

\begin{table}[h]
\centering
\caption{{Resource cost based on centralized hub entanglement allocation protocol, where $C(\rho)$ represents the classical communication cost defined in Eq.~\eqref{eq:c-cost}.} } 
\label{tab:resource}
\setlength{\tabcolsep}{1em}
\begin{tabular}{lccc}
\toprule
Allocated states & $C(\rho)$ &  \update{Memory cost} \\
\midrule
\addlinespace
$N$-qubit W states & $2N-2$ & {$N$}\\
\addlinespace
$N$-qubit GHZ states & $N$ & $N$\\
\addlinespace
$N$-qubit graph states & $N$ & $N$\\
\addlinespace
$N$-qubit pure states & $\geq N$ & {$N$} \\ 
\addlinespace
\bottomrule
\end{tabular}
\end{table}

It is worth noting that the circuit before measurement in Fig.~\ref{fig:n-w} can be further optimized using linear combination of unitaries~\cite{childs2012hamiltonian} and amplitude amplification~\cite{brassard1997exact,grover1998quantum} based on quantum singular value transformation techniques~\cite{gilyen2019quantum, martyn2021grand,Wang2022}. On the other hand, there also exist other different unitaries with shallow circuits at the expense of consuming more ancilla qubits. A detailed description of this method is provided in the Supplemental Material. We also remark that the protocols we discussed can achieve fidelity 1 under noiseless conditions. When encountering the noise situation, it is required to use entanglement distillation protocols on the user side. We refer~\cite{Dur2003purification} and~\cite{miguelramiro2023quantum} for purification of GHZ and W states, respectively. It will be interesting to explore machine learning enhanced LOCC protocols~\cite{Zhao2021} for distillation, dilution, and channel simulation protocols~\cite{Bennett1996,Fang2017,Wang2018g,Berta2015e} for multipartite entanglement.
{Ultimately, our findings introduce new methods for distributing quantum entanglement across centralized nodes, which could have significant applications in areas such as distributed quantum error correction \cite{xu2022distributed}, distributed quantum computing \cite{liu2023data,liu2024quantum}, and distributed quantum sensing tasks like quantum telescopes \cite{gottesman2012longer} within the context of large-scale quantum networks. }

\paragraph{Acknowledgements.---}
Y.-A. Chen and X. Liu contributed equally to this work.
We would like to thank Sumeet Khatri, Benchi Zhao, Yin Mo, Yu Gan, Hongshun Yao and Keming He for their helpful suggestions.
This work was supported by the National Key R\&D Program of China (Grant No.~2024YFE0102500), the Guangdong Provincial Quantum Science Strategic Initiative (Grant No.~GDZX2303007), the Guangdong Provincial Key Lab of Integrated Communication, Sensing and Computation for Ubiquitous Internet of Things (Grant No.~2023B1212010007),  the Start-up Fund (Grant No.~G0101000151) from HKUST (Guangzhou), the Quantum Science Center of Guangdong-Hong Kong-Macao Greater Bay Area, and the Education Bureau of Guangzhou Municipality.


\bibliography{references}


\clearpage
\appendix

\vspace{3cm}
\onecolumngrid
\vspace{2cm}

\begin{center}
\Large{\textbf{Supplementary Material}}
\end{center}

\renewcommand{\theequation}{S\arabic{equation}}
\numberwithin{figure}{section}
\renewcommand{\thesubsection}{\normalsize{Supplementary Note \arabic{subsection}}}
\renewcommand{\theproposition}{S\arabic{proposition}}
\renewcommand{\thedefinition}{S\arabic{definition}}
\renewcommand{\thefigure}{S\arabic{figure}}
\setcounter{equation}{0}
\setcounter{table}{0}
\setcounter{section}{0}
\setcounter{proposition}{0}
\setcounter{definition}{0}
\setcounter{figure}{0}

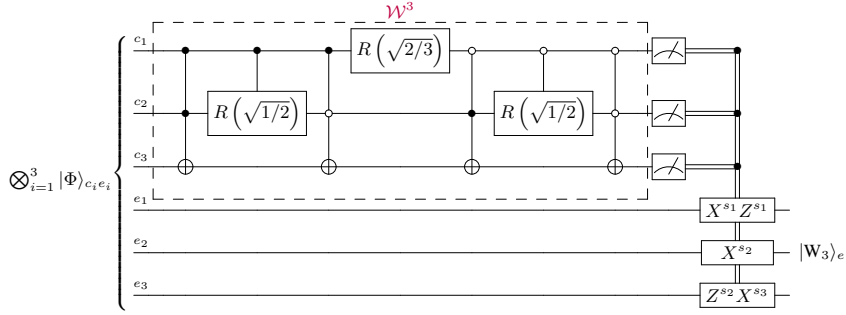
\begin{figure}[htbp]
\[
\resizebox{0.5\textwidth}{!}{
    \Qcircuit @C=0.8em @R=1em{
    & {} & {} & {} & {} &{}  & \mbox{\textcolor{purple}{$\wlo{3}$}} & {} & {} & {} \\
    & \qw_{\cnode{1}} & \qw &\ctrl{1} &\ctrl{1}  &\ctrl{1}  & \gate{R\left(\sqrt{2/3}\right)} &\ctrlo{1} &\ctrlo{1} & \ctrlo{1} & \qw  & \meter & \control \cw \cwx[1] \\
    & \qw_{\cnode{2}} & \qw &\ctrl{1} & \gate{R\left(\sqrt{1/2}\right)} &\ctrlo{1} & \qw  &\ctrl{1} & \gate{R\left(\sqrt{1/2}\right)} &\ctrlo{1} & \qw & \meter & \control \cw \cwx[1] \\
    & \qw_{\cnode{3}} & \qw &\targ & \qw  &\targ  & \qw &\targ & \qw & \targ & \qw  & \meter & \control \cw \cwx[1] \\
    & \qw_{\enode{1}} & \qw & \qw & \qw  & \qw  & \qw & \qw & \qw & \qw & \qw  & \qw & \gate{X^{s_1} Z^{s_1}} \cwx[1] & \qw \\
    & \qw_{\enode{2}} & \qw & \qw & \qw  & \qw  & \qw & \qw & \qw & \qw & \qw & \qw & \gate{\hspace{0.75em} X^{s_2} \hspace{0.75em}} \cwx[1] & \rstick{\w{3}_e} \qw  \\
    & \qw_{\enode{3}} & \qw & \qw & \qw  & \qw  & \qw & \qw & \qw & \qw & \qw & \qw & \gate{Z^{s_2} X^{s_3}} & \qw
    \gategroup{2}{4}{4}{10}{2.7em}{--}
    \inputgroupv{2}{7}{1.5em}{7em}{\bigotimes_{i=1}^3 \bellplus_{\cnode{i} \enode{i}} \hspace{4.5em}}
    }
}
\]
    \caption{Protocol for exact three-qubit W state distribution, using $3$ copies of a Bell state, where $R(x) \coloneqq R_y \left( 2 \arccos (x) \right)\cdot Z$. The local operation on each end node $\enode{i}$ depends on the measured outcome $s$, where $s= s_1 s_2 s_3 \in \set{0, 1}^{\times 3}$.}
    \label{fig:3-w}
\end{figure}

\section{Proof of the propositions about W states}\label{app:proof-w}

As a warm-up example, we first consider the case of $3$-qubit W states. Here we find a specific protocol for distributing 3-qubit W states, which can be implemented in three steps as follows (see Fig.~\ref{fig:3-w}).
\begin{enumerate}
    \item \textbf{Measure.} Evolve the state $\beginst = \bellplus_{\cnode{1} \enode{1}}\bellplus_{\cnode{2} \enode{2}}\bellplus_{\cnode{3} \enode{3}}$ on the central system $\csys = \cnode{1} \cnode{2} \cnode{3}$ by the following unitary
    \begin{align}\label{eq:w-u3}
       \wlo{3} = \frac{1}{\sqrt{3}} \bigl( 
       &X_{\cnode{1}} \ox I_{\cnode{2}} \ox I_{\cnode{3}} +Z_{\cnode{1}} \ox X_{\cnode{2}}  \ox I_{\cnode{3}} +Z_{\cnode{1}} \ox Z_{\cnode{2}} \ox X_{\cnode{3}} \bigr).\end{align}
    \item \textbf{Communicate.} Perform a two-outcome computational measurement 
    on each subsystem $c_i$ corresponding to the results ``0" and ``1", respectively. Then send the classical bits $s_1$, $s_2$, and $s_2$ with $s_3$ to three end nodes, respectively.
    
    \item \textbf{Recover.} Perform the local operation 
    \begin{align}
    \recover{s} = X^{s_1}Z^{s_1} \ox X^{s_2} \ox Z^{s_2} X^{s_3} 
    \end{align}
on the end system $\{e_i\}_{i=1}^3$ based on the received message.
\end{enumerate}

The protocol depicted in Fig.~\ref{fig:3-w} successfully allocates a $3$-qubit W state using 3 preshared Bell states and a communication cost of 4 classical bits. We summarize the above results into the following proposition with rigorous proof.

\begin{proposition}\label{lem:3-w}
 There exists a one-way LOCC protocol that deterministically and exactly distributes a $3$-qubit W state in a central hub with $3$ preshared Bell states and a classical communication cost of $4$ classical bits.
\end{proposition}
\begin{proof}
    For conveniences, denote $\beginst$ as the input state of systems $\csys, \esys$, and $U \coloneqq \wlo{3}$ as the local unitary acting on the central system $\csys$ in Fig.~\ref{fig:3-w}. Note that $U = \frac{1}{\sqrt{3}} \left(X \ox I \ox I + Z \ox X  \ox I + Z \ox Z \ox X \right)$ is symmetric. Then after receiving the measurement outcome $s = s_1 s_2 s_3$ on the central system, the corresponding output state $\finalst{s}$ is
\begin{align}
    \finalst{s} \coloneqq& \frac{P_s \left( U_\csys \ox I_\esys \right) \beginst }{ \norm{P_s \left( U_\csys \ox I_\esys \right) \beginst} } \textrm{, where } P_s = \bra{s}_\csys \ox I_\esys  \\
    =& \left( \bra{s}_\csys U_\csys \ox I_\esys \right) \sum_i \ket{i}_\csys \ox \ket{i}_\esys 
    = \sum_i \bra{s} U \ket{i} \ket{i} = \sum_i \bra{i} U^T \ket{s} \ket{i} = U^T \ket{s} \\
    =& \frac{1}{\sqrt{3}} \left[ \left(X \ket{s_1}\right) \ket{s_2} \ket{s_3} + (-1)^{s_1} \ket{s_1} \left(X \ket{s_2}\right) \ket{s_3} + (-1)^{s_1 + s_2} \ket{s_1} \ket{s_2} \left(X \ket{s_3}\right) \right]
.\end{align}
    At the last step, from the measurement result, one can correspondingly perform the recovery unitary $\recover{s} \coloneqq X^{s_1}Z^{s_1} \ox X^{s_2} \ox Z^{s_2} X^{s_3} $ on the end system, such that
\begin{alignat}{2}
    \recover{s} \finalst{s} &= \frac{1}{\sqrt{3}} [&&  X^{s_1} Z^{s_1} X \ket{s_1} \ox  X^{s_2} \ket{s_2} \ox  Z^{s_2} X^{s_3} \ket{s_3} + \\
    &&& (-1)^{s_1} X^{s_1}Z^{s_1} \ket{s_1} \ox X^{s_2 + 1} \ket{s_2} \ox Z^{s_2} X^{s_3} \ket{s_3} + \\
    &&& (-1)^{s_1 + s_2} X^{s_1}Z^{s_1} \ket{s_1} \ox X^{s_2} \ket{s_2} \ox Z^{s_2} X^{s_3} X \ket{s_3}] \\
    &= \frac{1}{\sqrt{3}} [&& (-1)^{s_1 + s_1^2} \ket{1} \ox \ket{0} \ox  \ket{0} + \\
    &&& (-1)^{s_1 + s_1^2} \ket{0} \ox \ket{1} \ox \ket{0} + \\
    &&& (-1)^{s_1 + s_2 + s_1^2} \ket{0} \ox \ket{0} \ox (-1)^{s_2} \ket{1}] \\
    &= \frac{1}{\sqrt{3}} [&& \ket{1} \ox \ket{0} \ox \ket{0} + \ket{0} \ox \ket{1} \ox \ket{0} + (-1)^{2s_2} \ket{0} \ox \ket{0} \ox \ket{1}] 
    = \w{3}
.\end{alignat}
We conclude that for all measurement outcome $s$, this protocol can exactly recover the output state of end system to $\w{3}$.

Based on the expression of the local operation $\recover{s}$, we observe that the first end node requires $s_1$, the second end node requires $s_2$, and the last end node requires both $s_2$ and $s_3$. Consequently, the communication cost amounts to $4$ classical bits.

\end{proof}

For the case of $N\geq 3$, Proposition~\ref{lem:3-w} can be extended to the case of $N$-qubit W states. We give the following lemma to assist the proof of Proposition~\ref{prop:n-w}.

\begin{lemma}\label{lem:xz-w}
For $s\in\{0,1\}$, we find 
\begin{align}
    &&(XZ)^sX\ket{s}=\ket{1},&&(XZ)^sZ\ket{s}=\ket{0};\\
    &&X^sI\ket{s}=\ket{0},&&X^sX\ket{s}=\ket{1}, &&X^sZ\ket{s}=(-1)^s\ket{0}.
\end{align}
\end{lemma}
\begin{proof}
It is checked that
\begin{align}
    &(XZ)^sX\ket{s}=(XZ)^s\ket{1-s}=(-1)^{s(1-s)}\ket{1}=\ket{1},\\
    &(XZ)^sZ\ket{s}=(-1)^{s}(XZ)^s\ket{s}=(-1)^{2s}\ket{0}=\ket{0},\\
    &X^sI\ket{s}=\ket{0},\\
    &X^sX\ket{s}=\ket{1},\\
    &X^sZ\ket{s}=(-1)^s X^s\ket{s}=(-1)^s\ket{0}.
\end{align}
\end{proof}

\renewcommand\theproposition{\ref{prop:n-w}}
\setcounter{proposition}{\arabic{proposition}-1}
\begin{proposition}[Protocol of W states allocation]
There exists a one-way LOCC protocol that deterministically and exactly distributes an $N$-qubit W state in a central hub with $N$ preshared Bell states and a classical communication cost of $2N-2$ classical bits.
\end{proposition}
\begin{proof}
For conveniences, denote $\beginst\coloneqq\ket{\Phi}_{c_ie_i}^{\ox N}$ as the input state of systems $\csys, \esys$, and $U \coloneqq \wlo{N}$ as the local unitary acting on the central system $\csys$ in Fig.~\ref{fig:n-w}. Note that 
\begin{equation}
    U=\frac{1}{\sqrt{N}}\sum_{r=1}^N Z^{\ox r-1}\ox X\ox I^{\ox N-r}=\frac{1}{\sqrt{N}}\sum_{r=1}^N\bigotimes_{k=1}^N(\delta_{k<r}Z+\delta_{k=r}X+\delta_{k>r}I),
\end{equation}
where
\begin{equation}
    \delta_p=\begin{cases}
        1,&p \text{ is true};\\
        0,&p \text{ is false}.
    \end{cases}
\end{equation}
It is trivial to prove that $U$ is unitary and real symmetric. Apply $U$ on central system, and we have
\begin{equation}\label{eq:UT}
\begin{aligned}
    &\left( \bra{s}_\csys \ox I_\esys \right)\left( U_\csys \ox I_\esys \right)\beginst
    =\frac{1}{\sqrt{2^N}}\left( \bra{s}_\csys U_\csys \ox I_\esys \right) \sum_i \ket{i}_\csys \ox \ket{i}_\esys \\
    =&\frac{1}{\sqrt{2^N}}\sum_i \bra{s}_\csys U_\csys \ket{i}_\csys \ox \ket{i}_\esys
    =\frac{1}{\sqrt{2^N}}\sum_i\ketbra{i}{i}U^T\ket{s}=\frac{U^T\ket{s}}{\sqrt{2^N}}=\frac{U\ket{s}}{\sqrt{2^N}},
\end{aligned}
\end{equation}
which implies that after receiving the measurement outcome $s=s_1s_2\cdots s_N$ on the end system, the corresponding post-measurement state $\finalst{s}$ is
\begin{align}
    \finalst{s} \coloneqq& \frac{\left( \bra{s}_\csys \ox I_\esys \right)\left( U_\csys \ox I_\esys \right)\beginst}{ \norm{\left( \bra{s}_\csys \ox I_\esys \right)\left( U_\csys \ox I_\esys \right)\beginst} }
    = \frac{U\ket{s}/\sqrt{2^N}}{\norm{U\ket{s}/\sqrt{2^N}}}
    =U\ket{s}.
\end{align}

At the last step, from the measurement result, one can correspondingly perform the recovery unitary on the end system. That is,
\begin{align}~\label{eq:v-w}
     \recover{s}=(XZ)^{s_1} \ox X^{s_2} \ox \bigotimes_{k=3}^NZ^{\sum_{l=2}^{k-1}s_l}X^{s_k}=(XZ)^{s_1} \ox\bigotimes_{k=2}^NZ^{\sum_{l=2}^{k-1}s_l}X^{s_k},
\end{align}
which satisfies
\begin{equation}
    \forall\, N\ge2,\ \recover{s}=V_{N-1}^{(s \operatorname{mod} 2^{N-1})}\ox Z^{\sum_{l=2}^{N-1}s_l}X^{s_N}.
\end{equation}
Then for any $s\in\{0,1\cdots,2^{N}-1\}$, by Lemma~\ref{lem:xz-w} we have
\begin{align}
    &\recover{s}U\ket{s}\\
    =&((XZ)^{s_1}\ox\bigotimes_{k=2}^nZ^{\sum_{l=2}^{k-1}s_l}X^{s_k})\cdot(\frac{1}{\sqrt{N}}\sum_{r=1}^N\bigotimes_{k=1}^N(\delta_{k<r}Z\ket{s_k}+\delta_{k=r}X\ket{s_k}+\delta_{k>r}I\ket{s_k})\\
    =&\frac{1}{\sqrt{N}}\sum_{r=1}^N \ket{\delta_{r=1}}\ox\bigotimes_{k=2}^N(\delta_{k<r}Z^{\sum_{l=2}^{k-1}s_l}X^{s_k}Z\ket{s_k}+\delta_{k=r}Z^{\sum_{l=2}^{k-1}s_l}X^{s_k}X\ket{s_k}+\delta_{k>r}Z^{\sum_{l=2}^{k-1}s_l}X^{s_k}I\ket{s_k})\\
    =&\frac{1}{\sqrt{N}}\sum_{r=1}^N \ket{\delta_{r=1}}\ox\bigotimes_{k=2}^N(\delta_{k<r}Z^{\sum_{l=2}^{k-1}s_l}(-1)^{s_k}\ket{0}+\delta_{k=r}Z^{\sum_{l=2}^{k-1}s_l}\ket{1}+\delta_{k>r}Z^{\sum_{l=2}^{k-1}s_l}\ket{0})\\
    =&\frac{1}{\sqrt{N}}\sum_{r=1}^N \ket{\delta_{r=1}}\ox\bigotimes_{k=2}^N(\delta_{k<r}(-1)^{s_k}\ket{0}+\delta_{k=r}(-1)^{\sum_{l=2}^{k-1}s_l}\ket{1}+\delta_{k>r}\ket{0})\\
    =&\frac{1}{\sqrt{N}}\sum_{r=1}^N \ket{\delta_{r=1}}\ox\bigotimes_{k=2}^N(\delta_{k<r}\ket{0}+\delta_{k=r}\ket{1}+\delta_{k>r}\ket{0})\\
    =&\frac{1}{\sqrt{N}}\sum_{r=1}^N \bigotimes_{k=1}^N\ket{\delta_{k=r}}\\
    =&{\w{N}}.
\end{align}
We conclude that for all measurement outcome $s$, this protocol can exactly recover the output state of end system to $\w{N}$.
Based on the expression of the local operation $\recover{s}$ in Eq.~\eqref{eq:v-w}, we observe that except the first end node requires $s_1$, the second end node requires $s_2$, and each subsequent $k$-th $(k\geq 3)$ end node requires both $s_k$ and $(\sum_{l=2}^{k-1}s_l\operatorname{mod} 2)$. Consequently, the communication cost amounts to $2N-2$ classical bits. 
\end{proof}

\begin{figure}[t]
\centering
\[
\Qcircuit @C=.7em @R=.8em{
    & {} & {} & {} & \mbox{\textcolor{purple}{$\ghzlo{N}$}}   \\
    & {} & {} &{}  \\
    & \lstick{\cnode{1}} & \qw & \ctrl{1} &\ctrl{3} &\gate{H} & \qw  &\meter& \control \cw \cwx[1]\\
    & \lstick{\cnode{2}} & \qw & \targ & \qw & \qw & \qw  &\meter& \control \cw \cwx[2]\\
    \vdots & & & &  &  & & \\
    & \lstick{\cnode{N}} & \qw & \qw & \targ & \qw & \qw  &\meter& \control \cw \cwx[1]\\
    & \lstick{\enode{1}} & \qw & \qw & \qw & \qw & \qw & \qw & \gate{\hspace{0.3em} Z^{s_1} \hspace{0.3em}} \cwx[1] & \qw \\
    & \lstick{\enode{2}} & \qw & \qw & \qw & \qw & \qw & \qw & \gate{\hspace{0.2em} X^{s_2} \hspace{0.2em}} \cwx[1] & \qw  \\
    & \lstick{\enode{3}} & \qw & \qw & \qw & \qw & \qw & \qw & \gate{\hspace{0.2em} X^{s_3} \hspace{0.2em}} \cwx[2] & \rstick{\ghz{N}_e} \qw  \\
    \vdots & & &  &  & &  & \vdots \\
    & \lstick{\enode{N}} & \qw & \qw & \qw & \qw & \qw & \qw & \gate{X^{s_N}} & \qw
    \gategroup{3}{4}{6}{6}{1.5em}{--}
    \inputgroupv{3}{11}{2.5em}{7em}{\bigotimes_{i=1}^N \bellplus_{\cnode{i} \enode{i}} \hspace{5em}}
}
\]
\caption{Protocols for distributing GHZ states for $N \geq 3$. In this setting, each end node $e_i$ preshares $N$ Bell states with the central node $c_i$, and the local operation on each end node $\enode{i}$ depends on the measured outcome $s = s_1 \cdots s_N \in \set{0, 1}^{\times N}$.}~\label{fig:n-ghz}
\end{figure}
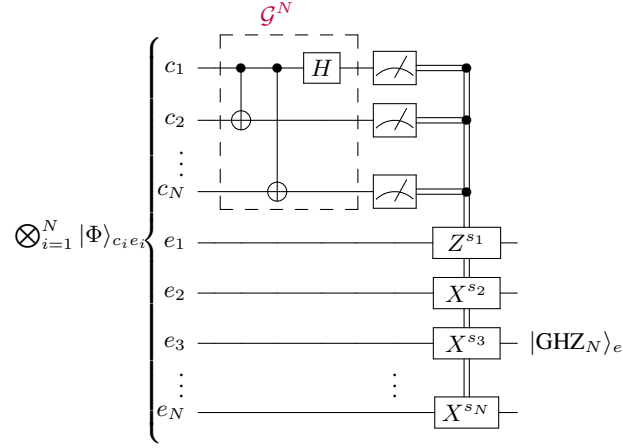

\section{Proof of the propositions about GHZ states}
\label{app:proof-ghz}

\renewcommand\theproposition{\ref{prop:n-ghz}}
\setcounter{proposition}{\arabic{proposition}-1}
\begin{proposition}[Protocol of GHZ states allocation]
There exists a one-way LOCC protocol that deterministically and exactly distributes an $N$-qubit GHZ state in a central hub with $N$ preshared Bell states and a classical communication cost of $N$ classical bits.
\end{proposition}

\begin{proof}
For conveniences, denote $\beginst\coloneqq\ket{\Phi}_{c_ie_i}^{\ox N}$ as the input state of systems $\csys, \esys$, and $U \coloneqq \ghzlo{N}$ as the local unitary acting on the central system $\csys$ in Fig.~\ref{fig:n-ghz}. Note that 
\begin{align}
    U&= \left(H \ox I^{ \ox (N-1)}\right)\cdot{\rm C}(X^{ \ox (N-1)})\\
    &= \left(H \ox I^{ \ox (N-1)}\right)\cdot \left(\ketbra{0}{0}\ox I^{ \ox (N-1)}+\ketbra{1}{1} \ox X^{ \ox (N-1)}\right),
\end{align}
which is clearly a unitary matrix. Apply $U$ on central system, and similarly as \eqref{eq:UT} we have
\begin{equation}\label{eq:UT2}
\begin{aligned}
    &\left( \bra{s}_\csys \ox I_\esys \right)\left( U_\csys \ox I_\esys \right)\beginst
    =\frac{1}{\sqrt{2^N}}\left( \bra{s}_\csys U_\csys \ox I_\esys \right) \sum_i \ket{i}_\csys \ox \ket{i}_\esys \\
    =&\frac{1}{\sqrt{2^N}}\sum_i \bra{s}_\csys U_\csys \ket{i}_\csys \ox \ket{i}_\esys
    =\frac{1}{\sqrt{2^N}}\sum_i\ketbra{i}{i}U^T\ket{s}=\frac{U^T\ket{s}}{\sqrt{2^N}},
\end{aligned}
\end{equation}
which implies that after receiving the measurement outcome $s=s_1s_2\cdots s_N$ on the end system, the corresponding post-measurement state $\finalst{s}$ is
\begin{align}
    \finalst{s} \coloneqq& \frac{\left( \bra{s}_\csys \ox I_\esys \right)\left( U_\csys \ox I_\esys \right)\beginst}{ \norm{\left( \bra{s}_\csys \ox I_\esys \right)\left( U_\csys \ox I_\esys \right)\beginst} }
    = \frac{U^T\ket{s}/\sqrt{2^N}}{\norm{U^T\ket{s}/\sqrt{2^N}}}
    =U^T\ket{s}.
\end{align}

Considering
\begin{align}
    U^T\ket{s}&={\rm C}(X^{ \ox (N-1)})\cdot (H \ox I^{ \ox (N-1)})\cdot ( X^{s_0} \ox X^{s_1} \ox\cdots \ox X^{s_{N-1}})\cdot\ket{0}^{ \ox N}\\
    &={\rm C}(X^{ \ox (N-1)})\cdot ( Z^{s_0} \ox X^{s_1} \ox\cdots \ox  X^{s_{N-1}})\cdot (H \ox I^{ \ox (N-1)})\cdot\ket{0}^{ \ox N}\\
    &=( Z^{s_0} \ox X^{s_1} \ox\cdots  \ox X^{s_{N-1}})\cdot {\rm C}(X^{ \ox (N-1)})\cdot (H \ox I^{ \ox (N-1)})\cdot\ket{0}^{ \ox N}\\
    &=(Z^{s_0} \ox X^{s_1} \ox \cdots \ox X^{ s_{N-1}})\ghz{N},
\end{align}
for each post-measurement state $U^T\ket{s}$, perform the operation $\recover{s}=Z^{s_0} \ox X^{s_1} \ox \cdots \ox X^{ s_{N-1}}$ on systems $\enode{1,\ldots, N}$, and we will obtain $\ghz{N}$. We conclude that for all measurement outcome $s$, this protocol can exactly recover the output state of the end system to $\ghz{N}$. Furthermore, we observe that each end nodes requires the corresponding $s_i$ to determine the local operation since the expression of $\recover{s}$. Thus, the communication cost amounts to $N$ classical bits. 

\end{proof}

\section{Optimality of the required classical bits}
\label{appen:opt}

In this section, we establish the optimality of the protocols distributing ${\w{N}}$ for $N\geq 3$ as demonstrated in Proposition \ref{prop:W3_4cbit} and Proposition \ref{prop:WN_2N-2cbit}. Additionally, we prove that the communication cost required for distributing any pure state using one-way LOCC in a central hub is at least $N$ classical bits, as outlined in Proposition \ref{prop:pure_state_cbit}, which implies the optimality of the protocol for distributing ${\ghz{N}}$.

To distributes an $N$-qubit state $\ket{\psi}$ using one-way LOCC 
\begin{equation}
    \cL = \sum_{s \in \cS} \cM^s_\csys \ox \cR^{1,\alpha_1(s)}_{\enode{1}} \ox \ldots \ox \cR^{N,\alpha_N(s)}_{\enode{N}}
.\end{equation}
in a central hub, it is equal to find $U\in\operatorname{SU}(2^N)$ and 
\begin{align}
\cR^{\alpha(s)}_{\enode{}}=\bigotimes_{k=1}^N\cR^{k,\alpha_k(s)}_{\enode{k}}, \cR^{k,\alpha_k(s)}\in\operatorname{SU}(2)
\end{align}
such that 
\begin{align}
    \forall\,s\in\{0,1,\ldots,2^N-1\},\cR^{\alpha(s)}\cM^s\cR^{\alpha(s)\dagger} =\ketbra{\psi}{\psi},\  \cM^s= U^T\ketbra{s}{s}U^*.
\end{align}
Then the classical communication cost for $\cL$ is defined as 
\begin{equation}
    \log_2\prod_{k=1}^N\max_s\alpha_k(s).
\end{equation}
To prove the protocol distributing ${\w{3}}$ is optimal in Proposition \ref{prop:W3_4cbit}, we need two lemmas firstly.
\renewcommand\theproposition{S1}
\begin{lemma}\label{lem:1cibt}
When distributing an $N$-qubit state $\ket{\psi}$ using one-way LOCC in a cental hub, if 
\begin{equation}
    \operatorname{Tr}_{\backslash l}\ketbra{\psi}{\psi}\neq I/2,
\end{equation}
then 
\begin{equation}
    \max_s\alpha_l(s)\ge2,
\end{equation}
or equivalently the $l$-th end requires at least $1$ classical bit, where $\operatorname{Tr}_{\backslash l}$ denotes the partial trace on all systems except the $l$-th system.
\end{lemma}
\begin{proof}
Suppose the $l$-th end does not require any classical bit, without loss of generality, we could assume that
\begin{align}
    &\exists\, U\in\operatorname{SU}(2^N),\forall\,s\in\{0,1,\ldots,2^N-1\},\exists\, \cR^{k,\alpha_k(s)}\in\operatorname{SU}(2),\\
    &\text{ s.t. } \cR^{\alpha(s)}_e=\bigotimes_{k=1}^{l-1} \cR^{k,\alpha_k(s)}_{e_k}\ox I_{e_l}\ox \bigotimes_{k=l+1}^N \cR^{k,\alpha_k(s)}_{e_k} \text{ satisfying } \cR^{\alpha(s)}U^T\ket{s}\propto\ket{\psi}.
\end{align}
Hence for any $s$,
\begin{equation}
    \operatorname{Tr}_{\backslash l}[U^T\ketbra{s}{s}U^*]=\operatorname{Tr}_{\backslash l}[(\cR^{\alpha(s)})^\dagger\ketbra{\psi}{\psi}\cR^{\alpha(s)}]=    \operatorname{Tr}_{\backslash l}[\ketbra{\psi}{\psi}],
\end{equation}
which makes a contradiction that
\begin{equation}
    2^{N-1}I=\operatorname{Tr}_{\backslash l}[U^TU^*]=\sum_{s=0}^{2^N-1}\operatorname{Tr}_{\backslash l}[U^T\ketbra{s}{s}U^*]=2^N\operatorname{Tr}_{\backslash l}[\ketbra{\psi}{\psi}]\ne2^{N-1}I.
\end{equation}
\end{proof}
\begin{corollary}\label{cor:W3-e1cbit}
When distributing an $N$-qubit W state using one-way LOCC in a central hub with $N>2$, each end node requires at least $1$ classical bit, or equivalently each $\max_s\alpha_l(s)\ge2$.
\end{corollary}
\begin{proof}
Considering 
\begin{equation}
    {\w{N}}=\frac{1}{\sqrt{N}}\ket{1}\ox\ket{0}^{\ox N-1}+\sqrt{\frac{N-1}{N}}\ket{0}\ox{\w{N-1}}
\end{equation}
we have each
\begin{equation}
    \operatorname{Tr}_{\backslash l}[\ketbra{\textup{W}_N}{\textup{W}_N}]=\frac{N-1}{N}\ketbra{0}{0}+\frac1N\ketbra{1}{1}\ne I/2\text{ for } N>2.
\end{equation}
Thus by the Lemma~\ref{lem:1cibt}, we are done.
\end{proof}

Now we could prove such optimality as following.
\begin{proposition}\label{prop:W3_4cbit}
The optimal classical communication cost for a $3$-qubit W state allocation is $4$ classical bits, regardless of the operations executed by the central system and end system of the central hub.
\end{proposition}
\begin{proof}
Suppose we could distribute an $3$-qubit W state using one-way LOCC in a cental hub within $3$ classical bits. By Corollary \ref{cor:W3-e1cbit}, each end requires at least $1$ classical bit. As a result, we could formalize the statement as 
\begin{equation}
    \exists\, U\in\operatorname{SU}(8),P,Q,R\in\operatorname{SU}(2),\forall\,s,\exists\, f_1,f_2,f_3\in\{0,1\},\text{ s.t. }(P^{f_1}\ox Q^{f_2}\ox R^{f_3})\cdot U^T\ket{s}\propto{\w{3}}.
\end{equation}
Considering
\begin{equation}
    \left\{U^T\ket{s}\right\}_j=\left\{(P^{-f_1}\ox Q^{-f_2}\ox R^{-f_3})\cdot{\w{3}}\,|\, f_1,f_2,f_3=0,1\right\}.
\end{equation}
is exactly a group of orthonormal basis of $\mathbb C^8$.
Since
\begin{equation}
    \bra{\textup{W}_3}\cdot(P^{-1}\ox I\ox I)\cdot{\w{3}}=0\Longleftrightarrow \bra{0}P\ket{0}=\bra{1}P\ket{1}=0,
\end{equation}
we have
\begin{equation}
    \bra{0}Q\ket{0}=\bra{1}Q\ket{1}=\bra{0}R\ket{0}=\bra{1}R\ket{1}=0,
\end{equation}
analogously.
By 
\begin{equation}
    {\w{3}}=\sqrt{\frac{2}{3}}{\w{2}}\ox\ket{0}+\sqrt{\frac{1}{3}}\ket{00}\ket{1},
\end{equation}
we have
\begin{align}
    0=&\bra{\textup{W}_3}\cdot(P^\dagger\ox Q^\dagger\ox I)\cdot{\w{3}}
    =\frac{2}{3}\bra{\textup{W}_2}(P^\dagger\ox Q^\dagger)\w{2}+\frac{1}{3}\bra{0}P^\dagger\ket{0}\bra{0}Q^\dagger\ket{0}\\
    =&\frac13(\bra{1}P^\dagger\ket{0}\cdot \bra{0}Q^\dagger\ket{1}+
    \bra{0}P^\dagger\ket{1}\cdot \bra{1}Q^\dagger\ket{0}),
\end{align}
and moreover $\operatorname{Tr}(P^\dagger Q^\dagger)=0$. Analogously, we find
\begin{equation}
    \operatorname{Tr}(P^\dagger Q^\dagger)=\operatorname{Tr}(P^\dagger R^\dagger)=\operatorname{Tr}(R^\dagger Q^\dagger)=0,
\end{equation}
which is contradictory with $P^\dagger,Q^\dagger,R^\dagger\in\operatorname{SU}(2)$ are all anti-diagonal.
\end{proof}

To extend Proposition \ref{prop:W3_4cbit} into cases on $\w{N}$ with $N\ge4$, we need more lemmas.
\renewcommand\theproposition{S2}
\begin{lemma}\label{lem:ortho_set}
There exists at most $2^{N-3}$ matrix $S_j\in\operatorname{SU}(2^{N-3})$, such that 
\begin{equation}
    \{(I_8\ox S_j)\cdot{\w{N}}\}_j
\end{equation}
is an orthogonal set.
\end{lemma}
\begin{proof}
Considering
\begin{equation}
    {\w{N}}=\sqrt{\frac{3}{N}}{\w{3}}\ox\ket{0}^{\ox N-3}+\sqrt{\frac{N-3}{N}}\ket{0}^{\ox 3}\ox{\w{N-3}},
\end{equation}
we find for $N>3$, 
\begin{align}
    &\dim_{\mathbb C}\operatorname{Span}(\{(I_8\ox S_j)\cdot{\w{N}}\}_j)
    =\dim_{\mathbb C}\operatorname{Span}({\w{3}}\ox\mathbb C^{2^{N-3}},\ket{000}\ox\mathbb C^{2^{N-3}})\\
    =&\dim_{\mathbb C}(({\w{3}}\oplus\ket{000})\ox\mathbb C^{2^{N-3}})=2^{N-2},
\end{align}
and 
\begin{equation}
    \{(I_8\ox S_j)\cdot{\w{N}}\}_j\subseteq \left(\sqrt{\frac{N-3}{N}}{\w{3}}-\sqrt{\frac{3}{N}}\ket{0}^{\ox 3})\ox \mathbb C^{2^{N-3}}\right)^\perp\cap\left(({\w{3}}\oplus\ket{000})\ox\mathbb C^{2^{N-3}}\right)
\end{equation}
with
\begin{align}
    &\dim_{\mathbb C}\left(\left(\sqrt{\frac{N-3}{N}}{\w{3}}-\sqrt{\frac{3}{N}}\ket{0}^{\ox 3})\ox \mathbb C^{2^{N-3}}\right)^\perp\cap\left(({\w{3}}\oplus\ket{000})\ox\mathbb C^{2^{N-3}}\right)\right)\nonumber\\
    =&2^{N-2}-2^{N-3}=2^{N-3},
\end{align}
As a result, there exists at most $2^{N-3}$ matrix $S_j\in\operatorname{SU}(2^{N-3})$, such that 
\begin{equation}
    \{(I_8\ox S_j)\cdot{\w{N}}\}_j
\end{equation}
is an orthogonal set.
\end{proof}

The following lemma plays a key role to derive a contradiction in proving the optimality of the protocols distributing ${\w{N}}$ for $N\ge4$.
\renewcommand\theproposition{S3}
\begin{lemma}\label{lem:WN->3}
The matrix equation 
\begin{equation}
    \sum_{s=0}^1\sum_{k=0}^1\sum_{l=0}^1(P^{-s}\ox Q^{-k}\ox R^{-l})\cdot\left(3\ketbra{\textup{W}_3}{\textup{W}_3}+(N-3)\ketbra{000}{000}\right)\cdot(P^{s}\ox Q^{k}\ox R^{l})= NI_8
\end{equation}
in variables $P,Q$ and $R$ has solutions in $\operatorname{SU}(2)$ only if $N=2$. 
\end{lemma}
\begin{proof}
Let $P=p_0I+ip_1X+ip_2Y+ip_3Z$, $Q=q_0I+iq_1X+iq_2Y+iq_3Z$, $R=r_0I+ir_1X+ir_2Y+ir_3Z$, and then we could find $N-2$ is in the Gr\"obner basis of the ideal generated by 
\begin{align}
    \Big\{&\sum_{s=0}^1\sum_{k=0}^1\sum_{l=0}^1(P^{-s}\ox Q^{-k}\ox R^{-l})\cdot\left(3\ketbra{\textup{W}_3}{\textup{W}_3}+(N-3)\ketbra{000}{000}\right)\cdot(P^{s}\ox Q^{k}\ox R^{l})- NI_8,\nonumber\\
    &\det P-1,\det Q-1,\det R-1\Big\}
\end{align}
in variables
\begin{equation}
    \{p_0,p_1,p_2,p_3,q_0,q_1,q_2,q_3,r_0,r_1,r_2,r_3,N\},
\end{equation}
which indicates that such matrix equation has solutions only if $N-2=0$. Here Gr\"obner basis is a key method to solving multivariate polynomial equation system in symbolic computation. Refers to \cite{cox2005using} to get more details for Gr\"obner basis.
\end{proof}

Now we could extend such optimality as following.
\renewcommand\theproposition{S2}
\begin{proposition}\label{prop:WN_2N-2cbit}
The optimal classical communication cost for an $N$-qubit W state allocation is $2N-2$ classical bits, regardless of the operations executed by the central system and end system of the central hub.
\end{proposition}
\begin{proof}
Nowadays the case for $N=3$ is proved in Proposition \ref{prop:W3_4cbit}, we will consider $N\ge4$ following.
Suppose we could distribute an $N$-qubit W state with three end nodes requiring only $1$ classical bit. By above lemmas, we have
\begin{align}
    &\exists\, U\in\operatorname{SU}(2^N),P_1,P_2,P_3,\cR^{k,\alpha_k(s)}\in\operatorname{SU}(2),f_1(s),f_2(s),f_3(s)\in\{0,1\},\\
    &\text{ s.t. }\forall\,s,\ \cR^{\alpha(s)}U^T\ket{s}\propto{\w{N}}\text{ and }\forall\,k=0,1,2,\ \cR^{k,\alpha_k(s)}=P_k^{f_k(s)}.
\end{align}
Considering
\begin{equation}
    U^T\ket{s}\propto(P_1^{f_1(s)}\ox P_2^{f_2(s)}\ox P_3^{f_3(s)}\ox \bigotimes_{k=4}^N\cR^{k,\alpha_k(s)})^\dagger{\w{N}},
\end{equation}
we find 
\begin{equation}
    \left\{(P_1^{f_1(s)}\ox P_2^{f_2(s)}\ox P_3^{f_3(s)}\ox \bigotimes_{k=4}^N\cR^{k,\alpha_k(s)})^\dagger{\w{N}}\right\}_s
\end{equation}
is exactly a group of orthonormal basis of $\mathbb C^{2^N}$. By Lemma \ref{lem:ortho_set} and the drawer principle, without loss of generality, we could assume $f_1(s)=s_1,f_2(s)=s_2,f_3(s)=s_3$,
i.e. each $P_1^{s_1}\ox P_2^{s_2}\ox P_3^{s_3}$ is corresponding to a $2^{N-3}$-dimensional subspace. Similar as the proof of Lemma \ref{lem:1cibt}, we find
\begin{align}
    &2^{N-3}I_8=\operatorname{Tr}_{\backslash\{1,2,3\}}[U^TU^*]
    =\sum_s\operatorname{Tr}_{\backslash\{1,2,3\}}[U^T\ketbra{s}{s}U^*]\\
    =&\sum_{s_1,s_2,s_3}2^{N-3}\operatorname{Tr}_{\backslash\{1,2,3\}}[(P_1^{-s_1}\ox P_2^{-s_2}\ox P_3^{-s_3})\ketbra{\textup{W}_N}{\textup{W}_N}(P_1^{s_1}\ox P_2^{s_2}\ox P_3^{s_3})]\\
    =&2^{N-3}\sum_{s_1,s_2,s_3}(P_1^{-s_1}\ox P_2^{-s_2}\ox P_3^{-s_3})\operatorname{Tr}_{\backslash\{1,2,3\}}[\ketbra{\textup{W}_N}{\textup{W}_N}](P_1^{s_1}\ox P_2^{s_2}\ox P_3^{s_3})\\
    =&\frac{2^{N-3}}{N}\sum_{s_1,s_2,s_3}(P_1^{-s_1}\ox P_2^{-s_2}\ox P_3^{-s_3})\left(3\ketbra{\textup{W}_3}{\textup{W}_3}+(N-3)\ketbra{000}{000}\right)(P_1^{s_1}\ox P_2^{s_2}\ox P_3^{s_3}),
\end{align}
which is contradictory with Lemma \ref{lem:WN->3}. Thus, when we distribute an $N$-qubit W state, except at most two end nodes, any other end nodes requires at least $2$ classical bits. As a conclusion, it requires at least $2N-2$ classical bits in this setting.
\end{proof}

Similarly, we find the communication cost of distributing any pure state is at least $N$ classical bits using one-way LOCC in a central hub in Proposition \ref{prop:pure_state_cbit}, which implies the optimality of the protocol distributing ${\ghz{N}}$.

\renewcommand\theproposition{S3}
\begin{proposition}\label{prop:pure_state_cbit}
The optimal classical communication cost for an arbitrary $N$-qubit pure state allocation is $N$ classical bits, regardless of the operations executed by the central system and end system of the central hub.
\end{proposition}
\begin{proof}
Considering $U^T\ket{s}=(\cR^{\alpha(s)})^\dagger\ket{\psi}$ is orthogonal to each other, it needs at least $N$ classical bits to distinguish those $\cR^{\alpha(s)}$s. As a result, it requires at least $N$ classical bits to distribute any $N$-qubit state using one-way LOCC in a central hub.
\end{proof}

\renewcommand\theproposition{S2}
\begin{corollary}
The protocol for allocating any $N$-qubit graph state in \cite {khatri2022design} achieves the optimal classical communication cost of $N$ classical bits. 
\end{corollary}
This result is obvious. The communication cost of any $N$-qubit graph state allocation protocol via one-way LOCC in a central hub is $N$ qubits~\cite{cuquet2012growth, khatri2022design}. Based on Proposition~\ref{prop:pure_state_cbit}, it is obvious that the optimal classical communication cost of any $N$-qubit graph state allocation protocol in \cite {khatri2022design} is $N$ classical bits.

\section{Methods to shallow circuits}
\label{appen:shallow}
In this section, we will discussion on how to shallow the circuit $\wlo{N}$ introduced to allocate W-state based on amplitude amplification.

Considering
\begin{equation}
    \wlo{N}=\frac{1}{\sqrt{N}}\sum_{s=0}^{N-1}Z^{\ox s}\ox X\ox I^{N-s-1}
\end{equation}
is exactly a linear combination of unitaries, we could use LCU method to implement $\wlo{N}$ intuitively. Specially, we could introduce another several ancilla qubits and quantum comparator to reduce the number of control gates significantly based on following decomposition.
\begin{align}
    \wlo{N}&=\frac{1}{\sqrt{N}}\sum_{s=0}^{N-1}Z^{\ox s}\ox X\ox I^{N-s-1}\\
    &=\frac{1}{\sqrt{N}}\sum_{s=0}^{N-1}\bigotimes_{k=0}^{N-1}(\delta_{s=k}X+\delta_{s<k}I+\delta_{s>k}Z)\\
    &=\frac{1}{\sqrt{N}}\sum_{s=0}^{N-1}\prod_{k=0}^{N-1}I_{2^{k}}\ox (\delta_{s=k}X+\delta_{s<k}I+\delta_{s>k}Z)\ox I_{2^{N-k-1}}
\end{align}
Denote $n=\lceil\log_2N\rceil$, and an implement for $n$-qubit full comparator with $m$ ancilla qubits as $P_{n,m}\in\operatorname{SU}(2^{2n+2+m})$, which satisfies
\begin{equation}
    P_{n,m}(\ket{a}_{2^n}\ket{b}_{2^n}\ox I_4\ox\ket{0}_{2^m})=\ket{a}_{2^n}\ket{b}_{2^n}\ox(\delta_{a=b}I_4+\delta_{a<b}I\ox X+\delta_{a>b}X\ox I)\ox\ket{0}_{2^m}.
\end{equation}
To simplify the notations, we omit the ancilla system $\ket{0}_{2^m}$ following, as
\begin{equation}
    P_{n}\ket{a}_{2^n}\ket{b}_{2^n}\ket{00}=\ket{a}_{2^n}\ket{b}_{2^n}(\delta_{a=b}\ket{00}+\delta_{a<b}\ket{01}+\delta_{a>b}\ket{10}).
\end{equation}
For $s=0\cdots,N-1$, denote $Q^s\in\operatorname{SU}(2^{N+2})$ satisfying
\begin{align}
    Q^s\cdot(\ket{00}\ox I_{2^N})&=\ket{00}\ox I_{2^{s}}\ox X\ox I_{2^{N-s-1}}\\
    Q^s\cdot(\ket{01}\ox I_{2^N})&=\ket{01}\ox I_{2^N}\\
    Q^s\cdot(\ket{10}\ox I_{2^N})&=\ket{10}\ox I_{2^{s}}\ox Z\ox I_{2^{N-s-1}}.
\end{align}

\begin{remark}
The following is an implement for $Q^s$ composed of only two two-qubit gates:
\begin{equation}
    (\ketbra{0}{0}\ox I\ox I_{2^N}+\ketbra{1}{1}\ox I\ox I_{2^{s}}\ox Z\ox I_{2^{N-s-1}})(I\ox \ketbra{0}{0}\ox I_{2^{s}}\ox X\ox I_{2^{N-s-1}}+I\ox\ketbra{1}{1}\ox I_{2^N}).
\end{equation}
\end{remark}
Denote $D_N$ as the preparation circuit for $\frac{1}{\sqrt{N}}\sum_{s=0}^{N-1}\ket{s}_{2^n}$, and thus for any $N$-qubit state $\ket{\psi}_{2^N}$, we have 
\begin{align}
&\ket{0}_{2^n}\ket{0}_{2^n}\ket{00}\ket{\psi}_{2^N}\\
\xrightarrow{D_N\ox I_{2^{n+2+N}}}
&\frac{1}{\sqrt{N}}\sum_{s=0}^{N-1}\ket{s}_{2^n}\ket{0}_{2^n}\ket{00}\ket{\psi}_{2^N}\\
\xrightarrow{P_n\ox I_{2^N}}
&\frac{1}{\sqrt{N}}\sum_j\ket{s}_{2^n}\ket{0}_{2^n}(\delta_{s=0}\ket{00}+\delta_{s<0}\ket{01}+\delta_{s>0}\ket{10})\ket{\psi}_{2^N}\\
\xrightarrow{I_{2^{2n}}\ox\recover{0}}
&\frac{1}{\sqrt{N}}\sum_j\ket{s}_{2^n}\ket{0}_{2^n}(((\delta_{s=0}\ket{00}X+\delta_{s<0}\ket{01}I+\delta_{s>0}\ket{10}Z)\ox I_{2^{N-1}})\ket{\psi}_{2^N})\\
\xrightarrow{P_n^\dagger\ox I_{2^N}}
&\frac{1}{\sqrt{N}}\sum_j\ket{s}_{2^n}\ket{0}_{2^n}\ket{00}(((\delta_{s=0}X+\delta_{s<0}I+\delta_{s>0}Z)\ox I_{2^{N-1}})\ket{\psi}_{2^N})\\
\rightarrow&\frac{1}{\sqrt{N}}\sum_j\ket{s}_{2^n}\ket{1}_{2^n}\ket{00}(((\delta_{s=0}X+\delta_{s<0}I+\delta_{s>0}Z)\ox I_{2^{N-1}})\ket{\psi}_{2^N}).
\end{align}
Similarly, by using $(P_n^\dagger\ox I_{2^N})(I_{2^{2n}}\ox \recover{k})(P_n\ox I_{2^N})$ we have
\begin{align}
\rightarrow
&\frac{1}{\sqrt{N}}\sum_j\ket{s}_{2^n}\ket{1}_{2^n}\ket{00}((I_2\ox(\delta_{s=1}X+\delta_{s<1}I+\delta_{s>1}Z)\ox I_{2^{N-2}})\nonumber\\
&\quad\cdot
(\delta_{s=0}X+\delta_{s<0}I+\delta_{s>0}Z)\ox I_{2^{N-1}})\ket{\psi}_{2^N})\\
\rightarrow
&\frac{1}{\sqrt{N}}\sum_j\ket{s}_{2^n}\ket{2}_{2^n}\ket{00}((I_2\ox(\delta_{s=1}X+\delta_{s<1}I+\delta_{s>1}Z)\ox I_{2^{N-2}})\nonumber\\
&\quad\cdot(\delta_{s=0}X+\delta_{s<0}I+\delta_{s>0}Z)\ox I_{2^{N-1}})\ket{\psi}_{2^N})\\
\xrightarrow{\cdots}
&\frac{1}{\sqrt{N}}\sum_j\ket{s}_{2^n}\ket{N-1}_{2^n}\ket{00}\left(\left(\prod_{k=0}^{N-1}I_{2^{k}}\ox (\delta_{s=k}X+\delta_{s<k}I+\delta_{s>k}Z)\ox I_{2^{N-k-1}}\right)\ket{\psi}_{2^N}\right)\\
\rightarrow
&\frac{1}{\sqrt{N}}\sum_j\ket{s}_{2^n}\ket{0}_{2^n}\ket{00}\left(\left(\prod_{k=0}^{N-1}I_{2^{k}}\ox (\delta_{s=k}X+\delta_{s<k}I+\delta_{s>k}Z)\ox I_{2^{N-k-1}}\right)\ket{\psi}_{2^N}\right)\\
\xrightarrow{D_n^\dagger\ox I_{2^{n+2+N}}}
&\frac{1}{N}\ket{0}_{2^n}\ket{0}_{2^n}\ket{00}\left(\sum_j\left(\prod_{k=0}^{N-1}I_{2^{k}}\ox (\delta_{s=k}X+\delta_{s<k}I+\delta_{s>k}Z)\ox I_{2^{N-k-1}}\right)\ket{\psi}_{2^N}\right)+*\ket{0^\perp}\\
=&\frac{1}{\sqrt{N}}\ket{0}_{2^n}\ket{0}_{2^n}\ket{00}\wlo{N}\ket{\psi}_{2^N}+\sqrt{\frac{N-1}{N}}\ket{0^\perp}.
\end{align}
Considering
\begin{equation}
    (\bra{0}_{2^n}\ox I_{2^{n+2+N}})\cdot\ket{0^\perp}=0,\ (I_{2^n}\ox\ketbra{0}{0}_{2^{n+2}}\ox I_{2^N})\cdot\ket{0^\perp}=\ket{0^\perp},
\end{equation}
afore circuit is a $\ketbra{0}{0}_{2^n}$-block-encoding of $\frac{1}{\sqrt{N}}\wlo{N}$, whose singular values are all equal to $\frac{1}{\sqrt{N}}$ as $\wlo{N}$ unitary. By using quantum singular value transformation(QSVT) satisfying $\frac{1}{\sqrt{N}}\mapsto 1$ and we could obtain a quantum circuit as a block-encoding of $\wlo{N}$, whose complexity is just those of afore circuit multiplied by $\cO(1/\arcsin\frac{1}{\sqrt{N}})=\cO(\sqrt{N})$.

Since the cost of other gates is negligible when comparing to quantum comparator and multi-control gates in the reflection in QSVT, the cost of preparing $\wlo{N}$ is exactly the cost of full quantum comparator and multi-control gates times $\cO(N^{1.5})$. When using quantum comparator of cost $\cO(\log N)$, the total cost for preparing $\wlo{N}$ is $O(N^{1.5}\cdot \log N)$, with ancilla qubit number $\cO(\log N)$.

\end{document}